\documentclass[12pt]{elsarticle}
\usepackage[margin=2.5cm]{geometry}
\usepackage{lipsum}

\makeatletter
\def\ps@pprintTitle{%
   \let\@oddhead\@empty
   \let\@evenhead\@empty
   \def\@oddfoot{\reset@font\hfil\thepage\hfil}
   \let\@evenfoot\@oddfoot
}
\makeatother

\usepackage{amsfonts,color,morefloats,pslatex}
\usepackage{amssymb,amsthm, amsmath,latexsym}

\newtheorem{remark}{Remark}
\newtheorem{theorem}{Theorem}
\newtheorem{lemma}[theorem]{Lemma}

\newtheorem{example}[theorem]{Example}

\newtheorem{conj}[theorem]{Conjecture}

\newcommand{\rank}{{\mathrm{rank}}}

\newcommand{\tr}{{\mathrm{Tr}}}

\newcommand{\gf}{{\mathbb{F}}}

\newcommand{\support}{{\mathrm{suppt}}}

\newcommand{\wt}{{\mathtt{wt}}}

\newcommand{\C}{{\mathcal{C}}}

\newcommand{\cC}{{\mathcal{C}}}

\newcommand{\bc}{{\mathbf{c}}}
\newcommand{\bg}{{\mathbf{g}}}

\usepackage{blindtext}
\begin{document}

\begin{frontmatter}



\title{Constructions of cyclic codes and extended primitive cyclic codes with their applications}

\tnotetext[fn1]{
This research was supported in part by the National Natural
Science Foundation of China under Grant 12271059, in part by
the Young Talent Fund of University Association for Science and
Technology in Shaanxi, China, under Grant 20200505, and in part
by the Fundamental Research Funds for the Central Universities,
CHD, under Grant 300102122202.
}

\author{Ziling Heng}
\ead{zilingheng@chd.edu.cn}
\author{Xinran Wang$^{\ast}$}
\ead{wangxr203@163.com}
\author{Xiaoru Li}
\ead{lixiaoru@163.com}

\cortext[cor]{Corresponding author}
\address{School of Science, Chang'an University, Xi'an 710064, China}





\begin{abstract}
Linear codes with a few weights have many nice applications including combinatorial design, distributed storage system, secret sharing schemes and so on.
In this paper, we construct two families of linear codes with a few weights based on special polynomials over finite fields. The first family of linear codes are extended primitive cyclic codes which are affine-invariant. The second family of linear codes are reducible cyclic codes. The parameters of these codes and their duals are determined.
  As the first application, we prove that these two families of linear codes hold $t$-designs, where $t=2,3$. As the second application, the minimum localities of the codes are also determined and optimal locally recoverable codes are derived.
\end{abstract}

\begin{keyword}
Linear code \sep cyclic code \sep extended primitive cyclic code

\MSC  94B05 \sep 94A05

\end{keyword}

\end{frontmatter}
\section{Introduction}
Let $\gf_q$ be the finite field with $q$ elements, where $q$ is a power of a prime. Let $\mathbb{F}_{q}^*:=\mathbb{F}_{q}\setminus \{0\}$. Let $\C$ be a non-empty set such that $\mathcal{C} \subseteq \mathbb{F}_{q}^{n}$. If $\C$ is a $k$-dimensional linear subspace over $\mathbb{F}_{q}$, then $\mathcal{C}$ is called an $[n,k,d]$ linear code over $\mathbb{F}_q$, where $d$ denotes its minimum distance. In particular, if any codeword $(c_0,c_1,\cdots,c_{n-1})\in \C$ implies $(c_{n-1},c_0,\cdots,c_{n-2})\in \C$, then
$\C$ is called a cyclic code.
The dual of an $[n,k,d]$ linear code $\C$ is defined by
$$
\mathcal{C}^{\perp}=\left\{ \textbf{u} \in \mathbb{F}_{q}^{n}: \langle  \textbf{u}, \mathbf{c} \rangle=0\mbox{ $\forall$ }\mathbf{c} \in \mathcal{C} \right\},
$$
where $\langle \textbf{u},\mathbf{c} \rangle$ denotes the standard inner product of $\textbf{u}$ and $\mathbf{c}$. It is obvious that $\mathcal{C}^{\perp}$ is an $[n,n-k]$ linear code.
Let $A_{i}$ denote the number of codewords with weight $i$ in a linear code of length $n$, where $0\leq i \leq n$. Then $A(z)=1+A_{1}z+A_{2}z^2+ \cdots +A_{n}z^n$ is referred to as the weight enumerator of $\mathcal{C}$. The sequence $(1,A_1,\cdots,A_n)$ is called the weight distributions of $\mathcal{C}$. The weight enumerator can be used not only to characterize the error detection and correction capabilities of linear codes, but also to calculate the error rate of error correction and detection. The weight enumerator of linear codes including cyclic codes has been studied in a large number of literatures in recent years \cite{T4,W0,H, HCXL, HLW, LYL1, W6, LYL, W3,W4}.

Let $\kappa,t$ and $n$ be positive integers with $1 \leq t \leq \kappa \leq n$. Let $\mathcal{P}$ be a set of $n$ elements and $\mathcal{B}$ be a set of $\kappa$-subsets of $\mathcal{P}$. The pair $\mathbb{D}=(\mathcal{P},\mathcal{B})$ is called a $t$-$(n,\kappa,\lambda)$ design, or simply $t$-design, if each $t$-subset of $\mathcal{P}$ is contained in precisely $\lambda$ elements of $\mathcal{B}$. The elements of $\mathcal{P}$ are called points and the elements of $\mathcal{B}$ are referred to as blocks. A $t$-design without repeated blocks is said to be simple. A $t$-design is called a Steiner system if $\lambda =1$ and $t \geq 2$, which is denoted by $S(t,\kappa,n)$. Some Steiner systems have been constructed  in \cite{T2, SD, SDT, WTD, 32D, SXT}.

Linear codes can be used to constructed $t$-designs. The well-known coding-theoretic construction is  described below. Let $\mathcal{P}=\{ 1,2,\cdots,n \}$ be a set of coordinate positions of the codewords of linear code $\mathcal{C}$ with length $n$. The support of a codeword $\textbf{c}=\{ c_1,c_2,\cdots,c_n \}$ in $\mathcal{C}$ is defined by $\support(\textbf{c})=\{1 \leq i \leq n : c_i \neq 0 \}$. Let $\mathcal{B}_\kappa$ denote the set of supports of all codewords with Hamming weight $\kappa$ in $\mathcal{C}$. The pair $(\mathcal{P},\mathcal{B}_\kappa)$ may be a $t$-$(n,\kappa,\lambda)$ design for some positive integer $\lambda$, which is referred to as a support design of $\mathcal{C}$. In other words, we say that the codewords with weight $\kappa$ in $\mathcal{C}$ support a $t$-$(n,\kappa,\lambda)$ design. When the pair $(\mathcal{P},\mathcal{B}_\kappa)$ is a simple $t$-$(n,\kappa,\lambda)$ design, we have the following relation:
\begin{eqnarray}\label{eqn-t}
|\mathcal{B}_\kappa|=\frac{1}{q-1}A_\kappa, \binom{n}{t}\lambda=\binom{\kappa}{t} \frac{1}{q-1}A_\kappa.
\end{eqnarray}

The following theorem developed by Assmus and Mattson gives a sufficient condition such that the pair $(\mathcal{P},\mathcal{B}_\kappa)$ defined in a linear code $\mathcal{C}$ is a $t$-design.

\begin{theorem}\label{the-AM}\cite{AM}
(Assmus-Mattson Theorem) Let $\mathcal{C}$ be an $[n,k,d]$ code over $\mathbb{F}_{q}$, and let $d^\perp$ denote the minimum distance of $\mathcal{C}^\perp$. Let $w$ be the largest integer satisfying $w \leq n$ and
\begin{eqnarray*}
w-\left\lfloor \frac{w+q-1}{q-2} \right\rfloor < d.
\end{eqnarray*}
Define $w^\perp$ analogously with $d^\perp$. Let $(A_0,A_1,\cdots,A_n)$ and $(A_0^\perp,A_1^\perp,\cdots,A_n^\perp)$ be the weight distributions of $\mathcal{C}$ and $\mathcal{C}^\perp$, respectively. Let $t$ be a positive integer with $t<d$ such that there are at most $d^\perp-t$ weights of $\mathcal{C}$ in the sequence $(A_0,A_1,\cdots,A_{n-t})$. Then
\begin{enumerate}
\item $(\mathcal{P},\mathcal{B}_\kappa)$ is a simple $t$-design provided that $A_\kappa \neq 0$ and $d \leq \kappa \leq w$;

\item $(\mathcal{P},\mathcal{B}_\kappa^\perp)$ is a simple $t$-design provided that $A_\kappa^\perp \neq 0$ and $d^\perp \leq \kappa \leq w^\perp$, where $\mathcal{B}_\kappa^\perp$ denotes the set of supports of all codewords of weight $\kappa$ in $\mathcal{C}^\perp$.
\end{enumerate}
\end{theorem}

The Assmus-Mattson Theorem is a powerful tool for construction $t$-design from linear codes \cite{T1,T2,T3,T4,T9,D,T5,T6,T7,32D,T8,T10}.

We can also use the automorphism group approach to obtain $t$-designs from linear codes. We review the automorphism group of linear codes for introducing this approach. The set of coordinate permutations that map a code $\cC$ to itself forms a group denoted by PAut($\cC$). PAut($\cC$) is called the permutation automorphism group of $\cC$. If $\cC$ is a code with length $n$, then PAut($\cC$) is a subgroup of the symmetric group Sym($n$). A monomial matrix over $\gf_q$ is a square matrix which has exactly one nonzero element of $\gf_q$ in each row and column. A monomial matrix $M$ can be written either in the form $PD$ or the form $D'P$ with $P$ a permutation matrix and $D$ and $D'$ being diagonal matrices. The set of monomial matrices that map $\cC$ to itself forms a group denoted as MAut($\cC$). MAut($\cC$) is called  the monomial automorphism group of $\cC$. It is obvious that PAut($\cC$) $\subseteq$ MAut($\cC$). The automorphism group Aut($\cC$) of $\cC$ is a set of maps with form $M\sigma$ that map $\cC$ to itself, where $M$ is a monomial matrix and $\sigma$ is a field automorphism. Then we have PAut($\cC$) $\subseteq$ MAut($\cC$) $\subseteq$ Aut($\cC$). Note that PAut($\cC$), MAut($\cC$) and Aut($\cC$) are the same in the binary case.

Clearly, every element in Aut($\cC$) has the form $DP\sigma$, where $D$ is a diagonal matrix, $P$ is a permutation matrix and $\sigma$ is an automorphism of $\gf_q$. If for every pair of $t$-element ordered sets of coordinates, there exists an element $DP\sigma$ in Aut($\cC$) such that its permutation part $P$ sends the first set to the second set, then  Aut($\cC$) is called $t$-transitive.
The following gives a sufficient condition for a linear code to hold $t$-designs.

\begin{theorem}\label{t-trans}\cite{FE}
Let $\cC$ be a linear code of length $n$ over $\gf_q$. If Aut($\cC$) is $t$-transitive,
then the codewords of any weight $i \geq t$ of $\cC$ hold a $t$-design.
\end{theorem}

The objective of this paper is to construct two families of linear codes with a few weights and study their applications. We first construct the linear codes based on special polynomials over finite fields. The first family of linear codes are extended primitive cyclic codes which are affine-invariant. The second family of linear codes are cyclic codes. The parameters of these codes and their duals are determined.
  As the first application, we prove that these two families of linear codes hold $t$-designs, where $t=2,3$. As the second application, the minimum localities of the codes are also determined and optimal locally recoverable codes are derived.

The remainder of this paper is organized as follows. In Section \ref{sec2}, we introduce some preliminary results on the number of zeros of some equations over finite fields and affine-invariant codes, which will be used in this paper. In Section \ref{sec3}, we construct a class of extended primitive cyclic codes by a special function and determine their parameters. We then derive some infinite families of $2$-designs and $3$-designs from these linear codes. In Section \ref{sec4}, we give another construction of linear codes, which are cyclic codes,  and determine their parameters and weight distributions. It turns out that they hold $3$-designs. In Section \ref{sec5}, we drive some optimal locally recoverable codes from these linear codes. In Section \ref{sec6}, we conclude the paper.

\section{Preliminaries}\label{sec2}

In this section, we will present some preliminary results on the number of zeros of some equations over $\gf_q$ and affine-invariant codes.

\subsection{The number of zeros of some equations over finite fields}

\begin{lemma}\label{root-p}
Let $h$ and $m$ be two integers with $h<m$ and let $q = p^m$ with $p$ a prime. Define a nonzero polynomial of the form
$$ g(x) = \sum_{i=0}^{h}a_{i}x^{p^i}, \ a_{i} \in \gf_q. $$
Denote  by $N_{g}$ the number of zeros of $g(x)$ in $\gf_q$. Then $N_{g} \in \{ 1, p, p^2, p^3, \cdots , p^h \}$.
\end{lemma}

\begin{proof}
It is obvious that $N_{g} \leq p^h $. Let $G$ be the set of zeros of $g(x)$ in $\gf_q$. Then $G\neq \emptyset$ as $0\in G$.
It is easy to prove that $(G,+)$ is a subgroup of $\gf_q$. By Lagrange's Theorem, the order of $G$ divides the order of $\gf_q$. Then we have $N_{g} \in \{ 1, p, p^2, p^3, \cdots , p^h \}$.
\end{proof}

In the following, we let $\alpha$ be a generator of $\gf_q^*$ and give some examples to verify Lemma \ref{root-p}.

\begin{example}\label{exa-1}
Let $p=2, m=5$ and $q=p^m$. Let
$$ g_1(x) = \alpha^2x + \alpha x^2 +\alpha^5x^4. $$
Denote by $N_{g_1}$ the number of zeros of $g_1(x)$ in $\gf_q$. By Magma program, $N_{g_1} = 2$.
\end{example}

\begin{example}\label{exa-2}
Let $p=2, m=4$ and $q=p^m$. Let
$$ g_2(x) = \alpha^3x + \alpha^5 x^2 +\alpha^8x^4 + \alpha^7 x^8.$$
Denote by $N_{g_2}$ the number of zeros of $g_2(x)$ in $\gf_q$. By Magma program, $N_{g_2} = 4$.
\end{example}

\begin{example}\label{exa-3}
Let $p=3, m=4$ and $q=p^m$. Let
$$ g_3(x) = \alpha^5x + \alpha^9 x^3 +\alpha^{12}x^9 + \alpha^{11} x^{27}.$$
Denote by $N_{g_3}$ the number of zeros of $g_3(x)$ in $\gf_q$. By Magma program, $N_{g_3}=9$.
\end{example}

\begin{example}\label{exa-4}
Let $p=2, m=4$ and $q=p^m$. Let
$$ g_4(x) = \alpha^{13}x + \alpha^7 x^2 +\alpha^{10}x^4 + \alpha x^8.$$
Denote by $N_{g_4}$ the number of zeros of $g_4(x)$ in $\gf_q$. By Magma program, $N_{g_4} = 8$.
\end{example}

\begin{example}\label{exa-5}
Let $p=3, m=3$ and $q=p^m$. Let
$$ g_5(x) = \alpha^{14}x + \alpha^{10} x^3 +\alpha^{24}x^9 .$$
Denote by $N_{g_5}$ the number of zeros of $g_5(x)$ in $\gf_q$. By Magma program, $N_{g_5} = 9$.
\end{example}

Let $h$ and $m$ be positive integers with $h<m$ and let $q=p^m$ with $p$ a prime. Now we consider the zeros of the nonzero polynomial
\begin{eqnarray}
f(x) = c + \sum_{i=0}^{h}a_{i}x^{p^i}, \ a_{i},c \in \gf_q,
\end{eqnarray}
in $\gf_q$. Let $g(x)$ be the polynomial defined in Lemma \ref{root-p}. It is obvious that $f(x)=g(x)+c, c \in \gf_q $.

\begin{lemma}\label{lem-0h}
Let $h$ and $m$ be positive integers with $h<m$ and let $q=p^m$ with p a prime. Denote by $N_{f}$ the number of zeros of $f(x)$ in $\gf_q$. Then $N_{f} \in \{0, 1, p, p^2, p^3, \cdots , p^h \}$.
\end{lemma}

\begin{proof}
If $a_i=0$ for all $0 \leq i \leq h$ and $c \neq 0$, then $N_f=0$.  Now we assume that $(a_0,a_1,\cdots,a_h)\neq (0,0,\cdots,0)$.
If $f(x)$ has a zero $u$ in $\gf_q$, then $f(u) = g(u) + c = 0$ which implies $c = -g(u)$. Thus $f(x)=g(x)+c=g(x)-g(u)=g(x-u)$. This implies that
$N_f=N_g$. By Lemma \ref{root-p}, we have $N_{f} \in \{1, p, p^2, p^3, \cdots , p^h \}$. The desired conclusion follows.
\end{proof}

\begin{lemma}\label{root}\cite{WTD}
Let $q = p^m$, where $p$ is an odd prime, $m \geq 2$. Let $1 \leq s \leq m-1$, $l=gcd(m,s)$. Let $U_{q+1}:=\{x \in \gf_{q^2}: x^{q+1} = 1\}$ and $f(x)=ax + bx^{p^s} + cx^{p^s+1} + u$, where $(a,b,c,u) \in \gf_{q^2}^4 \setminus \{0,0,0,0\}$. Then $f(x)$ has $0, 1, 2$ or $p^l+1$ zeros in $U_{q+1}$.
\end{lemma}

\subsection{Affine-invariant codes}\label{subsec2.2}
In this subsection, we introduce affine-invariant codes.

We give the definition of primitive cyclic codes at first. A primitive cyclic code is a cyclic code of length $n = q^m - 1$ over $\gf_q$, where $m$ is a positive integer. Let $R_n$ represent the quotient ring $\gf_q[x]/(x^n - 1)$. Any primitive cyclic code $\cC$ over $\gf_q$ is an ideal of $R_n$ which is generated by a monic polynomial $g(x)$ of the least degree over $\gf_q$. We call this polynomial the generator
polynomial of $\cC$. It can be represented as
$$
g(x)=\prod\limits_{t\in T}(x-\alpha^t),
$$
where $\alpha$ is a generator of $\gf_{q^m}^*$ and $T\subset \{0,1,\cdots,n-1\}$ is a union of some $q$-cyclotomic cosets modulo $n$.

We then introduce the extended primitive cyclic codes. Let $G$ denote the generator matrix of a primitive cyclic code $\mathcal{C}$. Define a matrix $\overline{G}$ by adding a column  to $G$ such that the sum of the elements of each row of $\overline{G}$ is $0$. The matrix $\overline{G}$ is the generator matrix of the extended code of $\mathcal{C}$. The extended code of a primite cyclic code $\mathcal{C}$ is called an extended primitive cyclic code and denoted by $\overline{\mathcal{C}}$.

Define the affine group $GA_1(\gf_q)$ by the set of all permutations $\sigma_{u,v}:x\longmapsto ux+v$ of $\gf_q$, where $u \in \gf_q^*$, $v \in \gf_q$. An affine-invariant code is an extended primitive cyclic code $\overline{\cC}$ such that $GA_1(\gf_q) \subseteq$ PAut($\overline{\cC}$).
It is easy to prove that the group action of $GA_1(\gf_q)$ on $\gf_q$ is doubly transitive, i.e. $2$-transitive. Then by the Theorem \ref{t-trans}, we have the following theorem.

\begin{theorem}\cite{FE}\label{AFF}
For each $i$ with $A_i \neq 0$ in an affine-invariant code $\overline{\cC}$, the supports of the codewords of weight $i$ form a $2$-design.
\end{theorem}

Theorem \ref{AFF} is a very useful tool in constructing $t$-designs from extended primitive cyclic  codes.


\section{A family of extended primitive cyclic codes}\label{sec3}
In this section, let $h$ and $m$ be positive integers with $h<m$ and $q=p^m$ with $p$ a prime. For convenience, let $\dim(\cC)$ and $d(\cC)$ respectively denote the dimension and minimum distance of a linear code $\cC$. Let $\alpha$ be a generator of $\gf_q^*$ and $\alpha_{i}:=\alpha^{i}$ for $1 \leq i \leq q-1$.

Define
\begin{eqnarray}\label{eqn-construction1}
D_{h}
=\left[
\begin{array}{ccccc}
1 & 1 & \cdots & 1  & 1\\
\alpha_1 & \alpha_2 & \cdots & \alpha_{q-1} & 0\\
\alpha_1^p & \alpha_2^p & \cdots & \alpha_{q-1}^p & 0\\
\alpha_1^{p^2} & \alpha_2^{p^2} & \cdots & \alpha_{q-1}^{p^2} & 0\\
 \vdots   & \vdots   & \vdots                &  \vdots     & \vdots          \\
\alpha_1^{p^h} & \alpha_2^{p^h} & \cdots & \alpha_{q-1}^{p^h} & 0
\end{array}
\right],
\end{eqnarray}
$D_h$ is an $h+2$ by $q$ matrix over $\gf_q$. Let $\C_{D_h}$ be the linear code over $\gf_q$ generated by $D_h$. Let $D_h'$ be the $h+2$ by $q-1$ submatrix of $D_h$ obtained by deleting the last column of $D_h$. Then the linear code $\C_{D_h'}$ generated by $D_h'$ is obvious a primitive cyclic code. It is easy to verify that $\C_{D_h}$  is the extended code of $\C_{D_h'}$.
Hence $\C_{D_h}$  is an  extended primitive cyclic code.

 In the following, we study the parameters of $\C_{D_h}$ and its dual $\C_{D_h}^\perp$ and obtain $t$-designs from them.

\begin{theorem}\label{thm-01}
Let $h$ and $m$ be positive integers with $h<m$ and let $q=2^m$. Then $\cC_{D_h}$ is a $[q, h+2, d]$ linear code with at most $h+2$ nonzero weights and $\C_{D_h}^\perp$ is a $[q, q-h-2, 4]$ linear code over $\gf_q$, where
$d \in \{ q-2^h, q-2^{h-1}, \cdots, q-2^j\}$ and $j$ is the least integer such that $2^j \geq h+1$. Moreover, $\C_{D_h}$ is affine-invariant and the supports of
all codewords of any  fixed nonzero weight in $\C_{D_h}$ form a $2$-design. Besides, the minimum weight
codewords of $\C_{D_h}^\perp$ support a $3$-$(q,4,1)$ simple design, i.e. a Steiner system $S(3,4,q)$.
\end{theorem}

\begin{proof}
We prove that $\dim(\cC_{D_h})=h+2$ at first. Let $\bg_i$, $1 \leqslant i \leqslant h+2$, represent the $i$-th row of $D_h$. Suppose that there are elements $r_i\in \gf_q$, $1 \leqslant i \leqslant h+2$, such that $\sum_{i=1}^{h+2}r_{i}\bg_i = 0$. Then
\begin{eqnarray*}
\left\{
\begin{array}{c}
r_1+r_2\alpha_1+r_2\alpha_1^{2}+ \cdots +r_{h+2}\alpha_1^{2^h}=0, \\
r_1+r_2\alpha_2+r_2\alpha_2^{2}+ \cdots +r_{h+2}\alpha_2^{2^h}=0, \\
          \vdots                 \\
r_1+r_2\alpha_{q-1}+r_2\alpha_{q-1}^{2}+ \cdots +r_{h+2}\alpha_{q-1}^{2^h}=0, \\
r_1=0.
\end{array}
\right.
\end{eqnarray*}
This implies that the polynomial $f(x)=r_1+r_2x+r_2x^{2}+ \cdots +r_{h+2}x^{2^h}$ has at least $q=2^m$ solutions.
Then $f(x)$ must be a zero polynomial as $h<m$.  In other words, we deduce that $r_i=0, 1 \leqslant i \leqslant h+2$ and $\bg_1$, $\bg_2, \cdots, \bg_{h+2}$ are linearly independent over $\gf_q$. Thus $\dim(\cC_{D_h})=h+2$.

We then prove that $\cC_{D_h}^\perp$ has parameters $[q, q-h-2, 4]$. Obviously, $\dim(\cC_{D_h}^\perp)=q-(h+2)=q-h-2$. Let $x_1, x_2, x_3$ be any three pairwise different elements in $\gf_q$. Consider the following submatrix of $D_{h}$ given by
\begin{eqnarray*}
M_{1}=
\left[
\begin{array}{lll}
1 & 1 & 1 \\
x_1 & x_2 & x_3 \\
\vdots   &      \vdots  & \vdots              \\
x_1^{2^h} & x_2^{2^h} & x_3^{2^h}
\end{array}
\right].
\end{eqnarray*}
Then we  consider the following submatrix of $M_{1}$ given as
\begin{eqnarray*}
M_{2}=
\left[
\begin{array}{lll}
1 & 1 & 1 \\
x_1 & x_2 & x_3 \\
x_1^{2} & x_2^{2} & x_3^{2}
\end{array}
\right].
\end{eqnarray*}
It is obvious that
$$|M_2| = \prod_{1 \leqslant i<j \leqslant 3}(x_j-x_i) \neq 0.$$
Then $\rank(M_{1})=3$ and any 3 columns of $D_h$ are linearly independent. This yields $d(\cC_{D_h}^\perp) \geq 4$. Let $x_1, x_2, x_3$ be any three pairwise different elements in $\gf_q$ and $x_4=x_1+x_2+x_3.$ Consider the following submatrix $M_3$ of $D_{h}$ given as
\begin{eqnarray*}
M_{3}=
\left[
\begin{array}{llll}
1 & 1 & 1 & 1\\
x_1 & x_2 & x_3 & x_4\\
\vdots   &  \vdots  & \vdots   & \vdots  \\
x_1^{2^h} & x_2^{2^h} & x_3^{2^h} & x_4^{2^h}
\end{array}
\right].
\end{eqnarray*}
Let $\bc_i$, $1 \leqslant i \leqslant 4$, represent the $i$-th column of $M_3$. It is obvious that $\bc_4 = \bc_1+\bc_2+\bc_3$, i.e. there exist four columns of $D_h$ which are linearly dependent. Hence, $d(\cC_{D_h}^\perp) = 4$ and $\C_{D_h}^\perp$ is a $[q, q-h-2, 4]$ code.

Now we determine the parameters of $ \cC_{D_h}$.  By definition, we have
$$ \cC_{D_h}=\{ \bc_{c,a_0,a_1,\cdots,a_{h}}: c,a_0,a_1,\cdots,a_{h} \in \gf_q \},$$
where
$$ \bc_{c,a_0,a_1,\cdots,a_{h}}=\left(c + \sum_{i=0}^{h}a_{i}x^{2^i}\right)_{x \in \gf_q}. $$
To determine the weight $\mbox{wt}(\bc_{c,a_0,a_1,\cdots,a_{h}})$ of a nonzero codeword $\bc_{c,a_0,a_1,\cdots,a_{h}} \in \cC_{D_h}$, it is sufficient to determine the number of zeros of the equation
$$ c + \sum_{i=0}^{h}a_{i}x^{2^i}=0$$
in $\gf_q$.
By Lemma \ref{lem-0h}, the above equation has $N_f$ zeros in $\gf_q$, where $N_{f} \in \{0, 1, 2, 4, 8, \cdots , 2^h \}$. Hence, $\mbox{wt}(\bc_{c,a_0,a_1,\cdots,a_{h}})\in \{q,q-1,q-2,q-4,q-8,\cdots,q-2^h\}$.
Then $q-2^h \leq d(\cC_{D_h}) \leq q-h-1$ by the Singleton bound. We then derive that $\cC_{D_h}$ is a $[q, h+2, d]$ code over $\gf_q$, where
$d \in \{ q-2^h, q-2^{h-1}, \cdots, q-2^j\}$, $j$ is the least integer such that $2^j \geq h+1$.

In what follows, we prove that $\cC_{D_h}$ is affine-invariant and holds 2-designs.
Let
$$ f(x):= c + \sum_{i=0}^{h}a_{i}x^{2^i},\ c,a_0,a_1,\cdots,a_h\in \gf_q.$$
For $u \in \gf_q^*, v \in \gf_q.$ We have
$$ f(ux+v) =  c + \sum_{i=0}^{h}a_{i}(ux+v)^{2^i} $$
$$ = c + \sum_{i=0}^{h}a_{i}v^{2^i} + \sum_{i=0}^{h}a_{i}u^{2^i}x^{2^i}$$
Let $\sigma_{(u,v)}(x) = ux+v \in GA_1(\gf_q)$ , where $u \in \gf_q^*$ and $ v \in \gf_q $. Then we have
$$ \sigma_{(u,v)}(\bc_{c,a_0,a_1,\cdots,a_{h}}) = \bc_{c',a_0',a_1',\cdots,a_{h}'} \in \cC_{D_h}, $$
where
$$ c' = f(v), a_i' = a_iu^{2^i} , 0 \leq i \leq h.$$
Note that $\cC_{D_h}$ is an extended primitive cyclic code by the definition in Subsection \ref{subsec2.2}. Then by the discussion above and the definition of PAut($\cC_{D_h}$), it is easy to deduce that $GA_1(\gf_q) \subseteq$ PAut($\cC_{D_h}$). Thus, $\cC_{D_h}$ is affine-invariant. By theorem \ref{AFF}, the supports of all codewords of any fixed nonzero weight in $\cC_{D_h}$ form a 2-design.

Finally, we prove that the minimum weight codewords of $\C_{D_h}^\perp$ support a 3-$(q,4,1)$ simple design, i.e. a Steiner system $S(3,4,q)$.
Taking any four different columns in $D_h$, we then obtain the matrix $M_3$. We now prove that $\rank(M_3)=3$ if and only if $x_4=x_1+x_2+x_3$.
It is obvious that $\rank(M_3)=3$ if $x_4=x_1+x_2+x_3$. Conversely, let $\rank(M_3)=3$ and we suppose $x_4 \neq x_1+x_2+x_3$. Consider the following submatrix of $M_3$ given as
\begin{eqnarray*}
M_{4}=
\left[
\begin{array}{llll}
1 & 1 & 1 & 1\\
x_1 & x_2 & x_3 & x_4\\
x_1^2 & x_2^2 & x_3^2 & x_4^2\\
x_1^4 & x_2^4 & x_3^4 & x_4^4
\end{array}
\right].
\end{eqnarray*}
Note that
$$ |M_{4}|= \prod_{1 \leqslant i<j \leqslant 4}(x_j-x_i)(x_1+x_2+x_3+x_4) \neq 0 $$
implying $\rank(M_3)=4$. This contradicts with $\rank(M_3)=3$.
Therefore, $\rank(M_3)=3$ if and only if $x_4=x_1+x_2+x_3$.
Let $\{x_{i_1}, x_{i_2}, x_{i_3}, x_{i_4}\}$ be any 4-subset of $\gf_q$ that satisfying $x_{i_4}=x_{i_1}+ x_{i_2}+x_{i_3}$, where $1 \leq i_j \leq q$.
Let $(r_{i_1}, r_{i_2}, r_{i_3}, r_{i_4})$ be a nonzero solution of
\begin{eqnarray*}
\left[
\begin{array}{llll}
1 & 1 & 1  & 1\\
x_{i_1} & x_{i_2} & x_{i_3} & x_{i_4}\\
x_{i_1}^2 & x_{i_2}^2 & x_{i_3}^2 & x_{i_4}^2\\
\vdots   &  \vdots  & \vdots   & \vdots  \\
x_{i_1}^{2^h} & x_{i_2}^{2^h} & x_{i_3}^{2^h} & x_{i_4}^{2^h}
\end{array}
\right]
\left[
\begin{array}{l}
r_{i_1}\\
r_{i_2}\\
r_{i_3}\\
r_{i_4}
\end{array}
\right]
=\mathbf{0}.
\end{eqnarray*}
Since rank$(M_{3})=3$, we have $r_{i_j} \neq 0$ for $1 \leq j \leq 4$. Let $\bc=(c_1,c_2, \ldots, c_q)$ be a codeword in $\cC_{D_h}^\perp$,
where $c_{i_j}=r_{i_j}$ and $c_v=0$ for all $v \in \{1,2, \ldots, q\} \setminus \{i_1, i_2, i_3, i_4\}$. It is clear that wt$(\bc)=4$.
Obviously, $\{a\bc:a \in \gf_q^* \}$ is a set of all codewords of weight 4 in $\cC_{D_h}^\perp$ whose nonzero coordinates are in the set $\{i_1, i_2, i_3, i_4 \}$. Therefore,
every codeword of weight 4 and its nonzero multiples in $\cC_{D_h}^\perp$ with nonzero coordinates $\{i_1, i_2, i_3, i_4 \}$ must correspond to the set $\{x_{i_1}, x_{i_2}, x_{i_3}, x_{i_4}\}$.
For every three pairwise distinct elements $x_{i_1}, x_{i_2}, x_{i_3}$ in $\gf_q$, the number of choices of $x_{i_4}$ is equal to 1 and independent of $x_{i_2}, x_{i_3}, x_{i_4}$. We then deduce that the codewords of weight 4
in $\cC_{D_h}^\perp$ support a 3-$(q,4,1)$ design. By Equation (\ref{eqn-t}), we have
$$ A_4^\perp=\frac{q(q-1)^2(q-2)}{24}.$$
The proof is completed.
\end{proof}

In Theorem \ref{thm-01}, the parameters of the $2$-designs derived from $\C_{D_h}$ are not given. It is open to determine them.
According to some examples confirmed by Magma program, we have the following conjecture.
\begin{conj}\label{conj}
The minimum weight codewords of $\C_{D_h}$ in Theorem \ref{thm-01} support $3$-designs.
\end{conj}

\begin{theorem}\label{thm-p2}
Let $p$ be an odd prime, $h$ and $m$ be positive integers with $h<m$ and $q=p^m$. Then $\cC_{D_h}$ is a $[q, h+2, d]$ code with at most $h+2$ nonzero weights and $\C_{D_h}^\perp$ is a $[q, q-h-2, 3]$ code over $\gf_q$, where
$d \in \{ q-p^h, q-p^{h-1}, \cdots, q-p^j\}$, $j$ is the least integer such that $p^j \geq h+1$. Moreover, $\C_{D_h}$ is affine-invariant and the supports of all codewords of any fixed nonzero weight in $\C_{D_h}$ form a 2-design. Besides, the minimum weight codewords of $\C_{D_h}^\perp$ support a $2$-$(q,3,p-2)$ simple design.
\end{theorem}

\begin{proof}
Similarly to Theorem \ref{thm-01}, we can easily derive the paraments of $\cC_{D_h}$. The possible nonzero weights of $\C_{D_h}$ are $\{ q-p^h, q-p^{h-1}, \cdots, q-p^j\}$, where $j$ is the least integer such that $p^j \geq h+1$.
Besides, we can also prove that $\C_{D_h}$ is affine-invariant.
Thus the supports of all codewords of any fixed nonzero weight in $\C_{D_h}$ form a 2-design.

In the following, we prove $\C_{D_h}^\perp$ has parameters $[q, q-h-2, 3]$.
Obviously, $\dim(\cC_{D_h}^\perp)=q-(h+2)=q-h-2$. It is obvious that any two columns are linearly independent in $D_{h}$, which implies $d(\cC_{D_h}^\perp) \geq 3.$ Let $x_1$ be an element in $\gf_q$ and $x_2=ax_1,a \in \gf_p^*\backslash \{1\}$. Consider the following submatrix of $D_{h}$ given as
\begin{eqnarray*}
M_{5}=
\left[
\begin{array}{lll}
1 & 1 & 1 \\
x_1 & x_2 & 0 \\
x_1^p & x_2^p & 0 \\
\vdots   &      \vdots  & \vdots              \\
x_1^{p^h} & x_2^{p^h} & 0
\end{array}
\right].
\end{eqnarray*}
Let $\bc_1$, $\bc_2$, $\bc_3$ represent the first, second, third column of $M_5$, respectively. It is easy to prove that $\bc_2 = a\bc_1 +(1-a)\bc_3$.  Hence, $d(\cC_{D_h}^\perp) = 3$ and $\C_{D_h}^\perp$ has parameters $[q, q-h-2, 3]$.

We now prove that the minimum weight codewords of $\C_{D_h}^\perp$ support a $2$-$(q,3,p-2)$ simple design.
Taking any three columns in $D_h$, we obtain the submatrix  $M_6$, where $x_1,x_2,x_3$ are pairwise distinct elements in $\gf_q$ and
\begin{eqnarray*}
M_{6}=
\left[
\begin{array}{lll}
1 & 1 & 1 \\
x_1 & x_2 & x_3 \\
x_1^p & x_2^p & x_3^p \\
\vdots   &      \vdots  & \vdots              \\
x_1^{p^h} & x_2^{p^h} & x_3^{p^h}
\end{array}
\right].
\end{eqnarray*}
Let $\bg_1$, $\bg_2$, $\bg_3$ represent the first, second, third column of $M_6$, respectively.
We first prove that $\rank(M_6)=2$ if and only if $x_3 = ax_1 +(1-a)x_2, a \in \mathbb{F}_{p} \setminus \{0,1\}$.
If $\bg_3=a\bg_1+(1-a)\bg_2$ if $x_3 = ax_1 +(1-a)x_2, a \in \mathbb{F}_{p} \setminus \{0,1\}$, then $\rank(M_6)=2$. Conversely, we let $\rank(M_6)=2$ and assume that $\bg_3=a\bg_1+b\bg_2, a,b \in \gf_q \setminus \{0\}$. Then we have
\begin{eqnarray}\label{x3}
\left\{
\begin{array}{c}
a+b=1, \\
ax_1+bx_2=x_3, \\
ax_1^p+bx_2^p=x_3^p, \\
          \vdots                 \\
ax_1^{p^h}+bx_2^{p^h}=x_3^{p^h}.
\end{array}
\right.
\end{eqnarray}
By the first two equations in  (\ref{x3}), we have $b=1-a$ and $x_3=ax_1 +(1-a)x_2, a \in \gf_q \setminus \{0,1\}$. Then by System (\ref{x3}) we have
\begin{eqnarray}\label{x3+}
\left\{
\begin{array}{c}
ax_1^p+(1-a)x_2^p=(ax_1+(1-a)x_2)^p, \\
ax_1^{p^2}+(1-a)x_2^{p^2}=(ax_1+(1-a)x_2)^{p^2}, \\
          \vdots                 \\
ax_1^{p^h}+(1-a)x_2^{p^h}=(ax_1+(1-a)x_2)^{p^h}.
\end{array}
\right.
\end{eqnarray}
where $a \in \gf_q \setminus \{0,1\}$. The System (\ref{x3+}) can be rewritten as
\begin{eqnarray*}\label{x3++}
ax_1^{p^i}+(1-a)x_2^{p^i}=(ax_1+(1-a)x_2)^{p^i},\ 1 \leq i \leq h,\ a \in \gf_q \setminus \{0,1\},
\end{eqnarray*}
which implies
\begin{eqnarray*}\label{a}
a^{p^i}(x_1-x_2)^{p^i}=a(x_1-x_2)^{p^i},\ 1 \leq i \leq h,\ a \in \gf_q \setminus \{0,1\}.
\end{eqnarray*}
Then $a^{p^i}=a$ for all $1 \leq i \leq h$. This implies
$$a \in \left(\bigcap_{i=1}^{h} \gf_{p^i}\right) \setminus \{0,1\}= \gf_p \setminus \{0,1\},$$
 and $x_3 = ax_1 +(1-a)x_2, a \in \mathbb{F}_{p} \setminus \{0,1\}.$ It is easy to prove that $x_3 \notin \{x_1,x_2\}$.Therefore, $\rank(M_6)=2$ if and only if $x_3 = ax_1 +(1-a)x_2, a \in \mathbb{F}_{p} \setminus \{0,1\}$.  If we fix $x_1,x_2$, then different choices of $a$ correspond to different $x_3$. Then the total number of different choices of $x_3$ such that $\rank(M_6)=2$ is equal to $p-2$. Let $x_{i_j}$ respectively denote the $i_j$-th column in $D_h$, where $1 \leq i_j \leq q$. Let $\{x_{i_1}, x_{i_2}, x_{i_3}\}$ be any 3-subset of $\gf_q$ that satisfying $x_{i_3}=ax_{i_1}+(1-a)x_{i_2}, a \in \mathbb{F}_{p} \setminus \{0,1\} $.
Let $(r_{i_1}, r_{i_2}, r_{i_3})$ be a nonzero solution of
\begin{eqnarray*}
\left[
\begin{array}{lll}
1 & 1 & 1  \\
x_{i_1} & x_{i_2} & x_{i_3} \\
x_{i_1}^p & x_{i_2}^p & x_{i_3}^p \\
\vdots   &  \vdots  & \vdots    \\
x_{i_1}^{p^h} & x_{i_2}^{p^h} & x_{i_3}^{p^h}
\end{array}
\right]
\left[
\begin{array}{l}
r_{i_1}\\
r_{i_2}\\
r_{i_3}
\end{array}
\right]
=\mathbf{0}.
\end{eqnarray*}
Since rank of the coefficient matrix equals $2$, then all $r_{i_j} \neq 0$ for $1 \leq j \leq 3$. Let $\bc=(c_1,c_2, \ldots, c_q)$ be a codeword in $\cC_{D_h}^\perp$,
where $c_{i_j}=r_{i_j}$ and $c_v=0$ for all $v \in \{1,2, \ldots, q\} \setminus \{i_1, i_2, i_3\}$. It is clear that wt$(\bc)=3$.
Obviously, $\{k\bc:k \in \gf_q^* \}$ is a set of all codewords of weight 3 in $\cC_{D_h}^\perp$ whose nonzero coordinates is $\{i_1, i_2, i_3\}$.
Therefore, every codeword of weight 3 and its nonzero multiples in $\cC_{D_h}^\perp$ with nonzero coordinates $\{i_1, i_2, i_3 \}$ must correspond to the set $\{x_{i_1}, x_{i_2}, x_{i_3}\}$. For every pair of distinct elements $x_{i_1}, x_{i_2}$ in $\gf_q$,
the number of different choices of $x_{i_3}$ is equal to $p-2$. We then deduce that the codewords of weight 3 in $\cC_{D_h}^\perp$ support a 2-$(q,3,p-2)$ design. By Equation (\ref{eqn-t}), we have
$$ A_3^\perp=\frac{q(q-1)^2(p-2)}{6}.$$
The proof is completed.
\end{proof}

\begin{remark}\label{rem}
By Theorems \ref{thm-01} and \ref{thm-p2}, we have $q-p^h \leq d(\cC_{D_h}) \leq q-p^j$, where $j$ is the least integer such that $p^j \geq h+1$. For $h=1,2,3,4$, we compute the parameters of $\cC_{D_h}$ and $\cC_{D_h}^\perp$ by magma in some cases. We list them in Table \ref{tab1}. These results show that the lower bound of $d(\cC_{D_h})$ is tight in these cases. It is open to determine the exact value of $d(\cC_{D_h})$ and the weight distribution of $\cC_{D_h}$ for general $h$.
\begin{table}[!htp]
\begin{center}
\caption{The parameters of $\cC_{D_h}$ and $\cC_{D_h}^\perp$ in Theorem \ref{thm-01}, \ref{thm-p2}. \label{tab1}}
\begin{tabular}{ccccc} \hline
 $p$ & $h$ & $m$  & $\cC_{D_h}$ &  $\cC_{D_h}^\perp$  \\ \hline
2&1  &2 & $[4,3,2]$& $[4,1,4]$ \\
2&1 &3 & $[8,3,6]$& $[8,5,4]$ \\
3&1 &3 & $[27,3,24]$& $[27,24,3]$ \\
5&1 &3 & $[125,3,120]$& $[125,122,3]$ \\
2&2 &3 & $[8,4,4]$& $[8,4,4]$  \\
2&2 &4  & $[16,4,12]$& $[16,12,4]$  \\
3&2 &3  & $[27,4,18]$& $[27,23,3]$  \\
5&2 &3  & $[125,4,100]$& $[125,121,3]$  \\
2&3 &4  & $[16,5,8]$& $[16,11,4]$  \\
2&3 &5  & $[32,5,24]$& $[32,27,4]$  \\
3&3 &4  & $[81,5,54]$& $[81,76,3]$  \\
2&4 &5  & $[32,6,16]$& $[32,26,4]$  \\
2&4 &6  & $[64,6,48]$& $[64,58,4]$  \\
3&4 &5  & $[243,6,162]$& $[243,237,3]$  \\
\hline
\end{tabular}
\end{center}
\end{table}
\end{remark}

In the following subsections, we determine the parameters  of $\cC_{D_h}$ for some special  $h$.
\subsection{When h =2}
If $p=h=2$, the weight distribution of $\cC_{D_h}$ was studied in \cite{32D}.  In this case, $\cC_{D_h}$ is an NMDS code holding $3$-designs.

\begin{theorem}\label{h=2}
Let $q = 2^m$ with $m >2$ and $h=2$. Then $\cC_{D_2}$ generated by the matrix $G_{D_2}$ is an NMDS code with parameters $[q, 4, q - 4]$ and weight enumerator
\begin{eqnarray*}
A(z)=1 + \frac{q(q-1)^2(q-2)}{24} z^{q-4} + \frac{q(q-1)^2(q+4)}{4} z^{q-2}+ \frac{q(q-1)(q^2+8)}{3} z^{q-1}+\\ \frac{(q-1)(3q^3+3q^2-6q+8)}{8} z^{q}.
\end{eqnarray*}
Moreover, the minimum weight codewords in $\cC_{D_2}$ support a $3$-$(q, q-4,\frac{(q-4)(q-5)(q-6)}{24})$ simple design and the minimum weight codewords in $\cC_{D_2}^{\perp}$ support a $3$-$(q, 4, 1)$ simple design, i.e., a Steiner system $S(3, 4, q)$. Furthermore, the codewords of weight $5$ in $\cC_{D_2}^{\perp}$ support a $3$-$(q, 5,\frac{(q-4)(q-8)}{2})$ simple design.
\end{theorem}

The weight enumerator of $\cC_{D_h}$ if $p>2, h=2$ is determined in the following theorem.

\begin{theorem}\label{h=2'}
Let $q=p^m$ with $p>2, m>2$ and $h=2$. Then $\cC_{D_h}$ is a $[q, 4, q-p^2]$ code over $\gf_q$ with weight enumerator
\begin{eqnarray*}
A(z)=1 + \frac{q(q-p)(q-1)^2}{p^3(p-1)^2(p+1)} z^{q-p^2} + \frac{q(q-1)^2(p^2q+p^2-q-pq)}{p^2(p-1)^2} z^{q-p}+ \\ \frac{q(q-1)(p^3q^2+p^3q+p^3-2p^2q^2-p^2q-pq^2-2pq+3q^2)}{(p-1)^2(p+1)} z^{q-1}+\\ \frac{(q-1)(p^3+p^2q^3-p^2q+pq^2-pq-q^3+q^2)}{p^3} z^{q}.
\end{eqnarray*}
\end{theorem}

\begin{proof}
By the proof Theorem \ref{thm-p2}, the possible nonzero weights of $\cC_{D_h}$ are $q,q-1,q-p,q-p^2$. Denote by $w_1=q, w_2=q-1, w_3=q-p, w_4=q-p^2.$ Let $A_{w_i}$
represent the frequency of the weight $w_i, 1 \leq i \leq 4$. By the first five Pless Power Moments in \cite{FE}, we have
\begin{eqnarray*}
\left\{
\begin{array}{ll}
\sum_{i=1}^{4}A_{w_i}=q^4-1, \\
\sum_{i=1}^{4}w_iA_{w_i}=q^4(q-1), \\
\sum_{i=1}^{4}w_i^2A_{w_i}=q^3(q^2-q+1)(q-1),\\
\sum_{i=1}^{4}w_i^3A_{w_i}=q[q(q-1)(q^4-2q^3+4q^2-4q+2)-6A_3^\perp],\\
\end{array}
\right.
\end{eqnarray*}
where $A_3^\perp$ is given in the proof of Theorem \ref{thm-p2}.
Solving the above system of linear equations gives
\begin{eqnarray*}
\left\{
\begin{array}{ll}
A_{q-p^2}=\frac{q(q-p)(q-1)^2}{p^3(p-1)^2(p+1)}, \\
A_{q-p}=\frac{q(q-1)^2(p^2q+p^2-q-pq)}{p^2(p-1)^2}, \\
A_{q-1}= \frac{q(q-1)(p^3q^2+p^3q+p^3-2p^2q^2-p^2q-pq^2-2pq+3q^2)}{(p-1)^2(p+1)} ,\\
A_{q}= \frac{(q-1)(p^3+p^2q^3-p^2q+pq^2-pq-q^3+q^2)}{p^3}.
\end{array}
\right.
\end{eqnarray*}
Then the weight enumerator of $\cC_{D_h}$ follows.
\end{proof}
\subsection{When h =3}

\begin{theorem}\label{thm-1}
Let $q=p^m$ with $p=2$, $m>3$ and $h=3$. Then $\cC_{D_h}$ is a $[q, 5, q-8]$ code over $\gf_q$ with weight enumerator
\begin{eqnarray*}
A(z)=1 + \frac{q(q-1)^2(q-2)(q-4)}{1344} z^{q-8} + \frac{q(q-1)^2(q-2)(3q+8)}{96} z^{q-4}+ \\ \frac{q(q-1)^2(7q^2+12q+32)}{24} z^{q-2}+\frac{2q(q-1)(3q^3+7q^2+32)}{21} z^{q-1}+ \\  \frac{(q-1)(25q^4+9q^3+22q^2-56q+64)}{64} z^{q}.
\end{eqnarray*}
\end{theorem}

\begin{proof}
By the proof of Theorem \ref{thm-01}, the possible weight of $\cC_{D_h}$ are $q,q-1,q-2,q-4,q-8$. Denote by $w_1=q, w_2=q-1,w_3=q-2, w_4=q-4, w_5=q-8.$ Let $A_{w_i}$
represent the frequency of the weight $w_i, 1 \leq i \leq 5$. By the first five Pless Power Moments in \cite{FE},  we have
\begin{eqnarray*}
\left\{
\begin{array}{ll}
\sum_{i=1}^{5}A_{w_i}=q^4-1, \\
\sum_{i=1}^{5}w_iA_{w_i}=q^4(q-1), \\
\sum_{i=1}^{5}w_i^2A_{w_i}=q^3(q^2-q+1)(q-1),\\
\sum_{i=1}^{5}w_i^3A_{w_i}=q^2(q-1)^2(q^4-2q^3+4q^2-4q+2),\\
\sum_{i=1}^{5}w_i^4A_{w_i}=q(q-1)(q^6-3q^5+9q^4-17q^3+22q^2-17q+6)+24A_4^\perp,
\end{array}
\right.
\end{eqnarray*}
where $A_4^\perp$ is determined in the proof of Theorem \ref{thm-01}.
Solving the above system of linear equations yields
\begin{eqnarray*}
\left\{
\begin{array}{ll}
A_{q-8}=\frac{q(q-1)^2(q-2)(q-4)}{1344}, \\
A_{q-4}=\frac{q(q-1)^2(q-2)(3q+8)}{96}, \\
A_{q-2}=\frac{q(q-1)^2(7q^2+12q+32)}{24},\\
A_{q-1}= \frac{2q(q-1)(3q^3+7q^2+32)}{21} ,\\
A_{q}= \frac{(q-1)(25q^4+9q^3+22q^2-56q+64)}{64}.
\end{array}
\right.
\end{eqnarray*}
Then the weight enumerator of $\cC_{D_h}$ follows.
\end{proof}

It is open to determine the weight enumerator of $\cC_{D_h}$ when $p>2$ and $h=3$.

\subsection{When $h\mid m$}
\begin{theorem}\label{th-hdivm}
Let $h$ and $m$ be positive integers with $h<m$, $h \mid m$ and $q=p^m$ with $p$ a prime.
Let $\C_{D_h}$ be the linear code over $\gf_q$ generated by $D_h$. Then $\C_{D_h}$ has parameters $[q, h+2, q-p^h]$.
\end{theorem}

\begin{proof}
Consider the polynomial
$$f(x) = c + \sum_{i=0}^{h}a_{i}x^{p^i}, \ a_{i},c \in \gf_q.$$
Let $a_{h}=1, a_{0}=-1, c=0, a_{i}=0, 1 \leq i \leq h-1$. Then $f(x)=x^{p^h} - x$. Let $f(x)=0$, then we have $x^{p^h}=x$, which implies $x \in \gf_{p^h} \subseteq \gf_q$ as $h \mid m$. Thus, the number of zeros of $f(x)$ in $\gf_q$ is equal to $p^h$. By the proofs of Theorems \ref{thm-01} and \ref{thm-p2},
$\C_{D_h}$ has parameters $[q, h+2, q-p^h]$ for any prime $p$.
\end{proof}

\subsection{When $h=m-1$}
The trace function from $\gf_q$ onto $\gf_p$ is defined by
 $$\tr_{q/p}(x)=x+x^{p}+x^{p^2}+\cdots+x^{p^{m-1}}.$$

\begin{theorem}\label{th-h=m-1}
Let $m$ be a positive integer  and $q=p^m$ with $p$ a prime.
Let $\C_{D_h}$ be the linear code over $\gf_q$ generated by $D_h$ and $h=m-1$. Then $\C_{D_h}$ has parameters $[q, m+1, q-p^{m-1}]$.
\end{theorem}

\begin{proof}
Let $h=m-1$. Consider the polynomial
$$f(x) = c + \sum_{i=0}^{h}a_{i}x^{p^i}, \ a_{i},c \in \gf_q.$$
Let $c=0, a_{i}=1, 0 \leq i \leq h$. Then $f(x)= \sum_{i=0}^{h}x^{p^i}= \tr_{q/p}(x)$. Let $f(x)=0$. Then the number of zeros of $f(x)$ in $\gf_q$ is equal to $p^h$. By the proofs of Theorems \ref{thm-01} and \ref{thm-p2}, the desired conclusion follows. 
\end{proof}

\section{A family of cyclic codes with four weights}\label{sec4}
Let $q$ be a power of an odd prime. Let $U_{q+1}:=\{x \in \gf_{q^2}: x^{q+1} = 1\}$. Let $x_{1}, x_{2}, \cdots, x_{q+1}$ denote all the elements of $U_{q+1}$. Let $G$ be a generator matrix of the linear code $\C$, where
\begin{eqnarray*}
G=\left[
\begin{array}{llll}
1 & 1 & \cdots & 1 \\
x_1 & x_2 & \cdots & x_{q+1} \\
x_1^{p^s} & x_2^{p^s} & \cdots & x_{q+1}^{p^s} \\
x_1^{p^s+1} & x_2^{p^s+1} & \cdots & x_{q+1}^{p^s+1}
\end{array}
\right].
\end{eqnarray*}
In fact, $\C$ is a reducible cyclic code as $U_{q+1}$ is a cyclic group.

\begin{theorem}\label{the-02}
Let $q = p^m$, where $p$ is an odd prime and $m \geq 2$. Let $1 \leq s \leq m-1$ and $l=\gcd(m,s)$. Then $\C$  is a $[q+1, 4, q-p^l]$ cyclic code with weight enumerator
\begin{eqnarray*}
A(z)=1 + \frac{(q+1)q(q-1)^2(p^l-q+p^lq^2+2p^lq^4+q^2+2q^4)}{2(p^l+1)} z^{q+1} + \\ \frac{(q+1)^2(q-1)(p^l-p^lq-q+p^lq^2-p^lq^3+p^lq^4+q^2)}{p^l} z^{q} + \\
\frac{(q+1)^2q(q-1)(p^l-q+p^lq^2-q^2)}{2(p^l-1)} z^{q-1} + \frac{(q+1)^2q(q-1)^2}{p^l(p^{2l}-1)} z^{q-p^l}.
\end{eqnarray*}
Moreover, $\C^\perp$ has parameters $[q+1, q-3, 4]$. The minimum weight codewords of $\C$ support a $3$-$(q+1,q-p^l,\frac{(q-p^l)(q-p^l-1)(q-p^l-2)}{p^l(p^{2l}-1)})$ simple design and
the minimum weight codewords of $\C^\perp$ support a $3$-$(q+1,4,p^l-2)$ simple design. When $p=3, l=1$, the minimum weight codewords of $\C^\perp$ support a $3$-$(3^m+1,4,1)$ simple design, i.e. a Steiner system $S(3,4,3^m+1)$.
\end{theorem}

\begin{proof}
We first prove that $\dim(\C)=4$. Let $\bg_1$, $\bg_2$, $\bg_3$, $\bg_4$ represent the first, second, third, forth row of $G$, respectively. Assume that $a\bg_1+b\bg_2+c\bg_3+u\bg_4=0$, then we have
\begin{eqnarray*}
\left\{
\begin{array}{c}
a+bx_1+cx_1^{p^s}+ux_1^{p^s+1}=0, \\
\vdots \\
a+bx_{q+1}+cx_{q+1}^{p^s}+ux_{q+1}^{p^s+1}=0. \\
\end{array}
\right.
\end{eqnarray*}
If $f(x)=a + bx + cx^{p^s} + ux^{p^s+1}$ is a nonzero polynomial, then it has at most  $p^l+1\leq p^{m-1}+1$ solutions in $U_{q+1}$. By the above System of equations,  we have $a=b=c=d=0$ and $\dim(\C)=4$.

We then prove that $\C^\perp$ has parameters $[q+1, q-3, 4]$. It is obviously that $\dim(\C^\perp)=q+1-4=q-3$. We prove that $d(\C^\perp)=4$ in the following.
Obviously, any two columns of $G$ are $\gf_{q^2}$-linearly independent. By the Singleton bound, we then have $3 \leq d(\C^\perp) \leq 5$.
Let $x,y,z$ be three pairwise different elements in $U_{q+1}$. We consider the following submatrix given by
\begin{eqnarray*}
D=\left[
\begin{array}{lll}
1 & 1 & 1 \\
x & y & z \\
x^{p^s} & y^{p^s} & z^{p^s} \\
x^{p^s+1} & y^{p^s+1} & z^{p^s+1}
\end{array}
\right].
\end{eqnarray*}
If $\frac{z-x}{y-x} \notin \gf_{p^s}^*$, we consider the submatrix $D_1$ of $D$, where
\begin{eqnarray*}
D_1=\left[
\begin{array}{lll}
1 & 1 & 1 \\
x & y & z \\
x^{p^s} & y^{p^s} & z^{p^s}
\end{array}
\right].
\end{eqnarray*}
Note that $|D_1|=(z-x)(y-x)^{p^s}-(y-x)(z-x)^{p^s}=0$ if and only if $(\frac{z-x}{y-x})^{p^s-1}=1$. Hence $|D_1|\neq0$ if $\frac{z-x}{y-x} \notin \gf_{p^s}^*$.
If $\frac{z-x}{y-x} \in \gf_{p^s}^*$, we consider the submatrix $D_2$ of $D$, where
\begin{eqnarray*}
D_2=\left[
\begin{array}{lll}
1 & 1 & 1 \\
x^{p^s} & y^{p^s} & z^{p^s} \\
x^{p^s+1} & y^{p^s+1} & z^{p^s+1}
\end{array}
\right].
\end{eqnarray*}
Suppose that $|D_2|=(y-x)(z-x)^{p^s}y^{p^s}-(z-x)(y-x)^{p^s}z^{p^s}=0$. Then $(\frac{z-x}{y-x})^{p^s-1}=(\frac{z}{y})^{p^s}$. Since $\frac{z-x}{y-x} \in \gf_{p^s}^*$, we have
$(\frac{z}{y})^{p^s}=1$. Then $\frac{z}{y}=1$ as $\gcd(p^s,q^2-1)=1$. This contradicts with $y\neq z$. Hence, $|D_2| \neq 0$.
We then deduce $\rank(D)=3$ and $4 \leq d(\C^\perp) \leq 5$. Now we prove that $\C$ has four possible nonzero weights $w_1=q+1, w_2=q, w_3=q-1, w_4=q-p^l$.
By definition, $$\C=\{c_{a,b,c,u}=(a+bx+cx^{p^s}+ux^{p^s+1})_{x \in U_{q+1}}:a,b,c,u \in \gf_{q^2}\}.$$  Note that the Hamming weight $\wt(c_{a,b,c,u})$ of $c_{a,b,c,u}$ satisfies $$\wt(c_{a,b,c,u}) \in \{q+1,q,q-1,q-p^l\}$$ by Lemma \ref{root}. If $d(\C^\perp)=5$, then $\C^\perp$ is a MDS code with parameters $[q+1,q-3,5]$ and $d(\C)=q-2$, which contradicts with $\wt(c_{a,b,c,u}) \in \{q+1,q,q-1,q-p^l\}$. Therefore, $d(\C^\perp)=4$ and $\C^\perp$ is an AMDS code with parameters $[q+1,q-3,4]$.

Finally, we calculate the weight enumerator of $\C$. Let $w_1=q+1, w_2=q,w_3=q-1, w_4=q-p^l.$ Let $A_{w_i}$
represent the frequency of the weight $w_i, 1 \leq i \leq 4$. Then by the first four Pless power moments in \cite{FE}, we have
\begin{eqnarray*}
\left\{
\begin{array}{l}
\sum_{i=1}^{4}A_{w_i} = (q^2)^4-1, \\
\sum_{i=1}^{4}w_iA_{w_i}= (q^2)^8(q^2n-n), \\
\sum_{i=1}^{4}w_i^2A_{w_i}= (q^2)^2[(q^2-1)n(q^2n-n+1)], \\
\sum_{i=1}^{4}w_i^3A_{w_i}= q^2[(q^2-1)n(q^4n^2-2q^2n^2+3q^2n-q^2+n^2-3n)]. \\
\end{array}
\right.
\end{eqnarray*}
Solving this system of linear equations yields
\begin{eqnarray*}
A_{w_1} &=& \frac{(q+1)q(q-1)^2(p^l-q+p^lq^2+2p^lq^4+q^2+2q^4)}{2(p^l+1)},\\
A_{w_2}&=&\frac{(q+1)^2(q-1)(p^l-p^lq-q+p^lq^2-p^lq^3+p^lq^4+q^2)}{p^l},\\
A_{w_3}&=&\frac{(q+1)^2q(q-1)(p^l-q+p^lq^2-q^2)}{2(p^l-1)},\\
 A_{w_4}&=&\frac{(q+1)^2q(q-1)^2}{p^l(p^{2l}-1)}.
\end{eqnarray*}
Then $\C$ has parameters $[q+1, 4, q-p^l]$ and the weight enumerator of $\C$ follows. By the Pless power moments in \cite{FE}, we have $A_4^\perp = \frac{(q+1)^2q(q-1)^2(p^l-2)}{24}$.
It follows from Theorem \ref{the-AM} and Equation (\ref{eqn-t}) that
the minimum weight codewords of $\C$ support a 3-$(q+1,q-p^l,\frac{(q-p^l)(q-p^l-1)(q-p^l-2)}{p^l(p^{2l}-1)})$ simple design and the minimum weight codewords of $\C^\perp$
support a 3-$(q+1,4,p^l-2)$ simple design.

The proof is completed.
\end{proof}

\begin{example}
Let $p=3, m=2, s=1$. Then the linear code $\C$  is an NMDS code  with parameters $[10,4,6]$ and weight enumerator
 $$A(z)=1+2400z^6+280800z^8+4743200z^9+38020320z^{10}.$$
The dual code $\C^{\perp}$ has parameters $[10,6,4]$. Besides, the codewords of weight $6$ in $\C$ support a $3$-$(10,6,5)$ simple design and the codewords of weight $4$ in $\C^\perp$ support a $3$-$(10,4,1)$ simple design, i.e. a Steiner system $S(3,4,10)$.
\end{example}

\begin{example}
Let $p=5, m=2, s=1$. Then the linear code $\C$ has parameters $[26,4,20]$ and weight enumerator
 $$A(z)=1+81120z^{20}+125736000z^{24}+6095697504z^{25}+146366376000z^{26}.$$
The dual code $\C^{\perp}$ has parameters $[26,22,4]$. Besides, the minimum weight codewords of $\C$ support a $3$-$(26,20,57)$ simple design and the minimum weight codewords of $\C^\perp$ support a $3$-$(26,4,3)$ simple design.
\end{example}

\section{Optimal locally recoverable codes}\label{sec5}
 Let $\mathcal{C}$ be a linear code with parameters $[n,k,d]$ over $\mathbb{F}_q$. For each positive integer $n$, let $[n]:=\{ 0,1,\cdots,n-1 \}$. Then we use the the elements in $[n]$ to index the coordinates of the codewords in $\mathcal{C}$.
For each $i \in [n]$, if there exist a subset $R_{i} \subseteq [n] \backslash {i}$ of size $r$ and a function $f_{i}(x_1,x_2,\cdots,x_r)$ on $\mathbb{F}_q^{r}$ meeting $c_i=f_{i}(\mathbf{c}_{R_i})$ for any $\mathbf{c}=(c_0,\cdots,c_{n-1}) \in \mathcal{C}$, then $\mathcal{C}$ is referred to as an $(n,k,d,q;r)$-LRC, where $\mathbf{c}_{R_i}$ is the projection of $\mathbf{c}$ at $R_{i}$. The set $R_{i}$ is known as  the repair set of $c_i$ and $r$ is called the locality of $\cC$. If each $f_i$ is a homogeneous function with degree $1$, then $\cC$ is called an $(n, k, d, q;r)$-LLRC (linearly local recoverable code) and has linear locality $r$. Obviously, each nontrivial linear code $\cC$ has a minimum linear locality.
The following lemma presents the relation between the minimum locality and the minimum linear locality of a nontrivial linear code.
\begin{lemma}[\cite{ZLRC}]
The minimum locality and minimum linear locality of a nontrivial linear code are equal.
\end{lemma}
Besides, there exist some tradeoffs among the parameters of LRCs. For each $(n,k,d,q;r)$-LRC, the Singleton-like bound (see \cite{CM}) is given as
\begin{eqnarray}\label{eqn-slbound}
d \leq n-k- \left \lceil \frac{k}{r} \right \rceil +2.
\end{eqnarray}
LRCs achieving this bound are said to be distance-optimal.
For any $(n,k,d,q;r)$-LRC, the Cadambe-Mazumdar bound (see \cite{GH}) is given by
\begin{eqnarray}\label{eqn-cmbound}
k \leq \mathop{\min}_{t \in \Bbb Z^{+}} [rt+k_{opt}^{(q)}(n-t(r+1),d)],
\end{eqnarray}
where $k_{opt}^{(q)}(n,d)$ represents the largest possible dimension of a linear code of length $n$ and minimum distance $d$ over $\gf_q$, and $\Bbb Z^{+}$ represents the set of all positive integers.
LRCs achieving the this bound are referred to as dimension-optimal ones.

The minimum locality of a nontrivial linear code $\cC$ is given as follows.

\begin{lemma}\cite{ZLRC}\label{locality}
Let $\cC$ be a nontrivial linear code of length $n$ and  $d^{\perp}=d(\cC^{\perp})$. The minimum locality of $\cC$ is $d^{\perp}-1$ if $(\mathcal{P}(\cC^{\perp}), \mathcal{B}_{d^{\perp}}(\cC^{\perp}))$ is a $1$-$(n,d^{\perp},\lambda_1^{\perp})$ design with $\lambda_1^{\perp} \geq 1$.
\end{lemma}

\begin{theorem}
Let $\cC_{D_h}$ be the code in Theorem \ref{thm-01} with $p=2$. Then $\cC_{D_h}$ is a $$\left(q, h+2, d, q; 3 \right)\mbox{-LRC}$$ and  $\cC_{D_h}^\perp$ is a $$\left(q, q-h-2,4,q;d-1\right)\mbox{-LRC,}$$
where $d \in \{ q-2^h, q-2^{h-1}, \cdots, q-2^j\}$ and $j$ is the least integer such that $2^j \geq h+1$.
\end{theorem}

\begin{proof}
The desired conclusion follows from  Theorem \ref{thm-01} and Lemma \ref{locality}.
\end{proof}

\begin{theorem}
Let $\cC_{D_h}$ be the code in Theorem \ref{thm-p2} with $p>2$. Then $\cC_{D_h}$ is a $$\left(q, h+2, d, q; 2 \right)\mbox{-LRC}$$ and  $\cC_{D_h}^\perp$ is a $$\left(q, q-h-2,3,q;d-1\right)\mbox{-LRC,}$$
where $d \in \{ q-p^h, q-p^{h-1}, \cdots, q-p^j\}$, $j$ is the least integer such that $p^j \geq h+1$.
\end{theorem}

\begin{proof}
The desired conclusion follows from  Theorem \ref{thm-p2} and Lemma \ref{locality}.
\end{proof}

\begin{theorem}\label{th-LLRC-0}
Let $\cC_{D_2}$ be the code in Theorem \ref{h=2} with $p=2,h=2$. Then $\cC_{D_2}$ is a $$\left(q, 4, q-4, q; 3 \right)\mbox{-LRC}$$ and  $\cC_{D_2}^\perp$ is a $$\left(q, q-4,4,q;q-5\right)\mbox{-LRC.}$$
Besides, $\cC_{D_2}$ and $\cC_{D_2}^\perp$ are both $d$-optimal and $k$-optimal.
\end{theorem}

\begin{proof}
By Theorem \ref{h=2} and Lemma \ref{locality}, the minimum localities of $\cC_{D_2}$ and $\cC_{D_2}^\perp$ are $d(\cC_{D_2}^\perp)-1=3 $ and $d(\cC_{D_2})-1=q-5$, respectively.
Then  $\cC_{D_2}$ is a
$$\left(q, 4, q-4, q; 3 \right)\mbox{-LRC}$$ and $\cC_{D_2}^\perp$ is a $$\left(q, q-4,4,q;q-5\right)\mbox{-LRC.}$$
We then prove $\cC_{D_2}$ is both $d$-optimal and $k$-optimal. By Equation (\ref{eqn-slbound}),
\begin{eqnarray*}
& &q-4-\left\lceil\frac{4}{3}\right\rceil+2\\
&=&q-4.
\end{eqnarray*}
Hence, $\cC_{D_2}$ is $d$-optimal.
Let $t=1$. Then
\begin{eqnarray*}
& &\min_{t\in \Bbb Z^+}\left[rt+k_{opt}^{(q)}\left(n-t(r+1),q-4\right)\right]\\
&=&3+k_{opt}^{(q)}(q-4,q-4)=4.
\end{eqnarray*}
Where the last equality holds due to $k_{opt}^{(q)}(q-4,q-4)=1$ by the Singleton bound.
By Equation (\ref{eqn-cmbound}), $\cC_{D_2}$ is $k$-optimal.
Similarly, we can prove $\cC_{D_2}^\perp$ is both $d$-optimal and $k$-optimal.
\end{proof}

Similarly, we can easily prove the following four theorems.

\begin{theorem}\label{th-LLRC-0'}
Let $\cC_{D_h}$ be the code in Theorem \ref{h=2'} with $h=2,p>2$. Then $\cC_{D_h}$ is a $$\left(q, 4, q-p^2, q; 2\right)\mbox{-LRC}$$ and  $\cC_{D_h}^\perp$ is a $$\left(q, q-4,3,q;q-p^2-1\right)\mbox{-LRC.}$$
Besides, $\cC_{D_h}^\perp$ is almost $d$-optimal.
\end{theorem}

\begin{theorem}\label{th-LLRC-1}
Let $\cC_{D_h}$ be the code in Theorem \ref{thm-1}, where $h=3$. Then $\cC_{D_h}$ is a $$\left(q, 5, q-8, q; 3 \right)\mbox{-LRC}$$ and  $\cC_{D_h}^\perp$ is a $$\left(q, q-5,4,q;q-9\right)\mbox{-LRC.}$$
Besides,$\cC_{D_h}^\perp$ is almost $d$-optimal.
\end{theorem}

\begin{theorem}\label{th-LLRC-1'}
Let $\cC_{D_h}$ be the code in Theorem \ref{th-hdivm} with $h\mid m$. 
\begin{enumerate}
\item If $p=2$, then $\cC_{D_h}$ is a $$\left(q, h+2, q-2^h, q; 3 \right)\mbox{-LRC}$$ and  $\cC_{D_h}^\perp$ is a $$\left(q, q-h-2, 4, q; q-2^h-1\right)\mbox{-LRC.}$$ Besides, when $h=1$ or $2$, $\cC_{D_h}$ and $\cC_{D_h}^\perp$ are both $d$-optimal and $k$-optimal. When $h=3$, $\cC_{D_h}^\perp$ is almost $d$-optimal.
\item If $p>2$, then $\cC_{D_h}$ is a $$\left(q, h+2, q-p^h, q; 2 \right)\mbox{-LRC}$$ and  $\cC_{D_h}^\perp$ is a $$\left(q, q-h-2, 3, q; q-p^h-1\right)\mbox{-LRC.}$$ Besides, when $p=3$, $h=1$, $\cC_{D_h}$ is both $d$-optimal and $k$-optimal, and $\cC_{D_h}^\perp$ is $k$-optimal. When $h=1$, $\cC_{D_h}^\perp$ is $d$-optimal. When $h=2$, $\cC_{D_h}^\perp$ is almost $d$-optimal.
\end{enumerate}
\end{theorem}

\begin{theorem}\label{th-LLRC-1''}
Let $\cC_{D_h}$ be the code in Theorem \ref{th-h=m-1} with $h=m-1$.
\begin{enumerate}
\item If $p=2$, then $\cC_{D_h}$ is a $$\left(q, m+1, q-2^{m-1}, q; 3 \right)\mbox{-LRC}$$ and  $\cC_{D_h}^\perp$ is a $$\left(q, q-m-1, 4, q; q-2^{m-1}-1\right)\mbox{-LRC.}$$ Besides, when $h=1$ or $2$, $\cC_{D_h}$ and $\cC_{D_h}^\perp$ are both $d$-optimal and $k$-optimal. When $h=3$, $\cC_{D_h}^\perp$ is almost $d$-optimal.
\item If $p>2$, then  $\cC_{D_h}$ is a $$\left(q, m+1, q-p^{m+1}, q; 2 \right)\mbox{-LRC}$$ and  $\cC_{D_h}^\perp$ is a $$\left(q, q-m-1, 3, q; q-p^{m-1}-1\right)\mbox{-LRC.}$$ Besides, when $p=3$, $h=1$, $\cC_{D_h}$ is both $d$-optimal and $k$-optimal, and $\cC_{D_h}^\perp$ is $k$-optimal. When $h=1$, $\cC_{D_h}^\perp$ is $d$-optimal. When $h=2$, $\cC_{D_h}^\perp$ is almost $d$-optimal.
\end{enumerate}
\end{theorem}

The minimum locality of $\C$ in Theorem \ref{the-02} is also studied in the following theorem.

\begin{theorem}\label{th-LLRC-2}
Let $\cC$ be the code in Theorem \ref{the-02}. Then $\cC$ is a $$\left(q+1, 4, q-p^l, q; 3 \right)\mbox{-LRC}$$ and  $\cC^\perp$ is a $$\left(q+1, q-3,4,q;q-p^l-1\right)\mbox{-LRC.}$$
Besides, $\cC$ is both $d$-optimal and $k$-optimal for $p=3,l=1$ and $\cC^\perp$ is both $d$-optimal and $k$-optimal for all odd prime $p$ and all $l=\gcd(m,s)$.
\end{theorem}

\begin{proof}
By Theorem \ref{the-02} and Lemma \ref{locality}, the minimum localities of $\cC$ and $\cC^\perp$ are $d(\cC^\perp)-1=3 $ and $d(\cC)-1=q-p^l-1$, respectively.
Then we directly derive that $\cC$ is a $$\left(q+1, 4, q-p^l, q; 3 \right)\mbox{-LRC}$$ and  $\cC^\perp$ is a $$\left(q+1, q-3,4,q;q-p^l-1\right)\mbox{-LRC.}$$

In the following, we prove $\C^\perp$ is both $d$-optimal and $k$-optimal.
\begin{eqnarray*}
& &q+1-(q-3)-\left\lceil\frac{q-3}{q-p^l-1}\right\rceil+2\\
&=&6-\left\lceil\frac{q-3}{q-p^l-1}\right\rceil=4.
\end{eqnarray*}
where the last equality holds due to $q-p^l-1 < q-3 \leq 2(q-p^l-1) $. Hence, $\C^\perp$ is  $d$-optimal by Equation (\ref{eqn-slbound}).
Let $t=1$. Then
\begin{eqnarray*}
& &\min_{t\in \Bbb Z^+}\left[rt+k_{opt}^{(q)}\left(n-t(r+1),4\right)\right]\\
&=&q-p^l-1+k_{opt}^{(q)}(p^l+1,4)=q-3.
\end{eqnarray*}
where the last equality holds due to $k_{opt}^{(q)}(p^l+1,4)=p^l-2$ by the Singleton bound.
Therefore, $\C^\perp$ is $k$-optimal by Equation (\ref{eqn-cmbound}).

Similarly, we can prove $\cC$ is both $d$-optimal and $k$-optimal when $p=3,l=1$. The proof is completed.
\end{proof}
\section{Summary and concluding remarks}\label{sec6}
In this paper, we constructed a family of extended primitive cyclic codes and a family of reducible cyclic codes by special polynomials. The parameters of them and their duals were determined. It was shown that these codes have nice applications in combinatorial designs and locally recoverable codes. Besides, a conjecture was given in Conjecture \ref{conj} and an open problem was proposed in Remark \ref{rem}. The reader is invited to solve them.




\end{document}